\documentclass{article}
\usepackage[utf8]{inputenc}
\usepackage{graphicx}
\usepackage{url}
\usepackage{appendix}
\usepackage{amssymb,amsmath,amsthm}
\usepackage{natbib}
\usepackage{color}
\usepackage{comment}
\usepackage{booktabs}
\usepackage[margin=1in]{geometry}
\usepackage{makecell}
\usepackage{amssymb} % for \checkmark
\usepackage{bbm}

\def\var{\text{var}}

\newtheorem{lemma}{Lemma}

\title{Bayesian Uncertainty Quantification for \\ Ranked Choice Voting Polls}
\author{
\begin{tabular}{c@{\hspace{2cm}}c}
Evan T.\ R.\ Rosenman\thanks{Co-first authors.} &
Jason Liang\footnotemark[1] \\
\\
\multicolumn{2}{c}{
\begin{tabular}{c}
\small Department of Mathematical Sciences\\
\small Claremont McKenna College\\
\small Claremont, CA, United States
\end{tabular}
}
\end{tabular}
}
\begin{document}

\maketitle

\begin{abstract}
Ranked choice voting (RCV) is a popular alternative voting method in which voters are asked to list their favored candidates in preference order, rather than vote for a single candidate. When these ballots are tabulated, candidates are successively eliminated, and their votes are reallocated to each voter's next-preferred choice. The process continues until a single candidate commands a majority of the active ballots and is declared the winner. As RCV gains wider adoption, the method poses novel challenges for pollsters. Unlike traditional, plurality elections, the event that a given candidate wins cannot be expressed in terms of a single population parameter. Hence, the basic concept of a margin-of-error is no longer straightforward to define. Moreover, a candidate’s ability to win may depend on both their support across the ballot and the order by which other candidates are eliminated. Existing measures of sampling uncertainty for polls of RCV elections do not clearly quantify these path-dependent outcomes. 

Here, we propose a simple, Bayesian framework to quantify uncertainty in polls of RCV elections. We cast the problem as one of estimating win probabilities for each leading candidate, and leverage a simple conjugacy relationship to estimate these probabilities conditional on the poll results. We also include applied analyses involving two prominent ranked choice voting elections: the 2021 New York City Democratic mayoral primary, in which Eric Adams narrowly defeated Kathryn Garcia in the final round; and the 2022 special election to Alaska's U.S. House seat, in which Mary Peltola was elected despite not being a Condorcet winner. Using the cast vote records from both elections, we demonstrate some of the challenges of traditional frequentist uncertainty quantification in RCV polls. We also demonstrate the utility of our approach using a poll of the NYC primary obtained from the polling firm Data for Progress. 
\end{abstract}

\tableofcontents

\section{Introduction}

Ranked choice voting (RCV) is an increasingly popular electoral system, in which voters are asked to rank their favored candidates in order rather than select a single choice. Voters can express more complex preferences using RCV, and its advocates frequently tout the potential to improve electoral representation and responsiveness \citep{vishwanath2024effects}. The method has been widely adopted in the United States, and is now used for all federal and state-level elections in Alaska and all federal elections in Maine, as well as municipal elections and primaries in cities such as New York City, San Francisco, Portland, and Minneapolis. 

In tabulations for ranked choice voting elections, ballots are initially assigned to their first-ranked candidate. If no candidate receives a majority of the first-choice votes, the candidate with the least first-choice support is eliminated and the candidates ranked after the eliminated candidates on any ballot each move up one spot. For example, in tabulations following the elimination of a voter’s top-ranked candidate, their second-ranked candidate becomes top-ranked, their third-ranked candidate becomes second-ranked, and so on. In each round of eliminations, the highest-ranked non-eliminated candidate becomes the first choice of each ballot. Once all candidates on a ballot have been eliminated, that ballot is considered “exhausted” and no longer counts toward the vote totals in future rounds. The rounds of elimination and retallying continue until one candidate secures a majority of the active ballots. See \cite{PaG3955} for a detailed explanation of RCV and some variations. 

RCV elections differ from traditional, plurality elections in that the election winner's victory depends on both how frequently candidates are ranked at the top of ballots \emph{and} the relative rankings of different frontrunners across ballots.  Statistically, this means that RCV election outcomes cannot be expressed as a function of a single population parameter. Moreover, unlike in plurality elections, the event that a candidate wins the election is a union of several disjoint pathways to victory. Hence, traditional frequentist tools for summarizing polling uncertainty, such as the margin-of-error, cannot be straightforwardly adapted to RCV elections. This poses a methodological challenge for survey analysts as they seek to report inferential results for polls conducted on RCV elections. 

In this paper, we propose a simple Bayesian framework for quantifying uncertainty in ranked choice voting polls. We formulate the problem in terms of ``ballot orderings," the unique possible ways a voter can rank a set of candidates. For example, in an RCV election with candidates A, B, and C, a voter may rank A first, B second, and C third; rank only A; or rank B followed by C. Each option corresponds to a unique ballot ordering. The observed frequencies of these orderings are sufficient statistics for the underlying population proportions, which in turn uniquely determine the RCV winner. Moreover, we can naturally model the voters as independent and identically distributed draws from a multinomial distribution over the set of possible ballot orderings. Using a well-known conjugacy relationship, we can obtain a posterior Dirichlet distribution over the ballot orderings when conditioning on the polling data, even if we default to a flat prior over ballot orderings. Sampling from the posterior gives a straightforward way to estimate the probability that each candidate triumphs in the eventual election.  

Throughout, we consider two recent elections which demonstrate the path-dependent nature of candidate elimination -- and thus the challenges of uncertainty estimation -- in ranked choice voting polls. The first election is the June 22, 2021 Democratic mayoral primary in New York City, in which Brooklyn Borough President Eric Adams won the final round by a narrow 0.8\% margin. Many pre-election polls had suggested his competitor in the final round would be civil rights attorney Maya Wiley, whom Adams would have defeated comfortably, but Wiley was actually eliminated in the second-to-last round. Instead, former NYC Sanitation Commissioner Kathryn Garcia advanced to the final round and nearly defeated Adams. The second election is the August 16, 2022 U.S. House special in Alaska, which resulted in the election of Democrat Mary Peltola to Congress. The election exhibits an extreme form of path dependency: the first-eliminated candidate, Nick Begich III, actually would have defeated Peltola in a head-to-head election. 

For both elections, we use the cast vote record (CVR) -- electronic records that list the full set of rankings for each voter who cast a ballot -- to simulate unbiased polls \citep{kuriwaki2024cast, shirani2018disentangling}. For the New York City race, we also obtain individual-level response data from the final poll conducted by the outlet Data for Progress, which correctly identified that Garcia would be highly competitive with Adams if she made it to the final round \citep{dfp2021nycrcvmemo}. 

The remainder of this paper proceeds as follows. Section \ref{sec:litReview} provides relevant background, including a literature review, notation and assumptions, and an introduction to the datasets for Alaska and NYC elections. Section \ref{sec:challenges} considers the challenges in using standard frequentist tools to assess uncertainty in ranked choice voting polls. Section \ref{sec:method} introduces our method, and Section \ref{sec:app} demonstrates its utility on the Alaska and NYC elections. Section \ref{sec:discussion} concludes. 

\section{Background}\label{sec:litReview} 

\subsection{Notation \& Assumptions}\label{sec:assumptions}

In the sections to follow, we suppose we sample $i = 1, \dots, n$ individuals from a population of $N$ total voters who will participate in an election, where typically $n \ll N$. We will assume the polled units are a true random sample from the voter population, meaning every individual in the population of voters has an equal probability of being sampled into the poll. In practice, pollsters frequently assign survey weights $w_i$ to individuals in their samples to account for unequal sampling probabilities. Such weights are utilized in the Data for Progress poll described in the next section. When using survey weights, we assume that the $w_i$ are correctly specified such that weighting can unbiasedly recover the distribution of preferences in the underlying voter population.

We denote as $y_i$ the response recorded from polled individual $i$. For simplicity, we will ignore ``Don't know" or ``Unsure" responses, and suppose that the poll comprises only respondents who state a voting intention. In the case of a two-candidate, first-past-the-post election, we can encode $y_i \in \{0, 1\}$ by choosing $y_i = 1$ to denote support for one of the two candidates and $y_i = 0$ to denote support for the other candidate. 

By contrast, in the case of RCV elections, we encode $y_i \in \mathcal{P}$, where $\mathcal{P}$ is the finite set of all possible ballot orderings among the set of candidates, $\mathcal{C}$. For example, in an RCV election with three candidates in which voters may rank up to three choices, $|\mathcal{P}| = 15$. Denoting the candidates as $A, B,$ and $C$ for simplicity,
\[ \mathcal{P} = \{p_1,\ldots,p_{15}\} = \{ A, B, C, AB, AC, BA, BC, CA, CB, ABC, ACB, BAC, BCA, CAB, CBA\}. \]
More generally, given a pool of $K$ candidates and allowing voters to rank up to $R$ candidates so that $R \leq K$, the number of unique ballot orderings is
\[ |\mathcal{P}| = \sum_{r=1}^{R} \frac{K!}{(K-r)!}. \]
Lastly, let
\[ \boldsymbol{\pi} = (\pi_1,\ldots,\pi_{|\mathcal P|}) \]
denote the population distribution of ballot orderings, where
\[ \pi_j = \Pr(y_i = p_j). \]
The vector $\boldsymbol{\pi}$ completely characterizes the distribution of voter preferences and uniquely determines the outcome of the RCV election (in the absence of exact ties). 

Throughout the manuscript, we will interchangeably use the terms ``first-past-the-post (FPTP)" and ``plurality" elections to mean those elections in which voters cast ballots for a single candidate and the candidate with the largest number of votes wins.  

\subsection{Standard FPTP Polling Practices} \label{sec:stan}

We briefly discuss common practices for uncertainty quantification in standard, two-candidate FPTP elections. Let $\pi_A$ and $\pi_B$ denote the true proportions of voters who support candidates $A$ and $B$, respectively, and define the election margin as 

\[ d_{AB} = \pi_A - \pi_B. \]
Because $\pi_A + \pi_B = 1$, the outcome of the election is determined by the sign of $d_{AB}$: candidate $A$ wins if $d_{AB} > 0$, while candidate $B$ wins if $d_{AB} < 0$.

Consequently, uncertainty quantification in a two-candidate FPTP poll reduces to inference on a single scalar parameter. Pollsters can estimate $\hat \pi_A$ and $\hat \pi_B$ based on the proportion of survey respondents who state a preference for each candidate, and can then estimate $\hat d_{AB} = \hat \pi_A - \hat \pi_B$. Under our assumptions (as well as some weak regularity conditions), $\hat d_{AB}$ is unbiased and approximately normally distributed, with variance
\[ \var\left(\hat d_{AB}\right) = \frac{1}{n} \bigg(\pi_A(1 - \pi_A) + \pi_B(1 - \pi_B) + 2\pi_A\pi_B\bigg).\]
Because the quantity in parentheses is bounded above by $1$, with equality achieved at $\pi_A=\pi_B=1/2$, it is standard to replace the variance by the conservative estimate $1/n$.

Under these assumptions, pollsters can straightforwardly report an $\alpha$-level confidence interval for the margin,
\[ \hat d_{AB} \pm z_{1-\alpha/2} \times \frac{1}{\sqrt{n}},\] 
where $z_{1-\alpha/2}$ is the $1 - \alpha/2$ quantile of a standard normal distribution. The ``margin of error" -- i.e. half the length of the confidence interval -- is also typically reported. Because the problem is governed by a single margin, winner determination can also be framed as testing
\[ H_0 : d_{AB} \le 0 \qquad \text{vs.} \qquad H_1 : d_{AB} > 0. \]
Under the normal approximation, the corresponding one-sided test is uniformly most powerful. For technical details underlying these standard practices, see Appendix \ref{app:fptp}.

\subsection{Related Literature and Current Practices}\label{sec:currPractice}

Much of the existing literature on ranked choice voting considers the implications of using the method as an electoral reform. An important series of papers considers the differential impacts of adoption of RCV along dimensions such as age, ethnicity, and socioeconomic status. Some work considers the direct impact on voter participation  \citep{mcdaniel2016writing, DowlingTolbertMicatkaDonovan2024Turnout}, while others make use of cast vote records to empirically evaluate how voters fill out their ballots \citep{kuriwaki2024cast}. Some researchers have found differential burdens of reform across key groups, though these claims are highly contested in the literature \citep{BurnettKogan2014, Pettigrew2023Overvotes, cormack2026more, maloy2021impact, coll2021demographic}. Other work has considered the impact of RCV adoption on political polarization and candidate behavior \citep{atkinson2024beyond, kropf2021using}. 

Social choice theorists, mathematicians, and computer scientists have also shown deep interest in ranked choice voting. Because RCV outcomes depend on the full structure of voters' preference rankings, it exhibits several counterintuitive properties. These include the previously-mentioned failure to obey the Condorcet criterion, sensitivity to elimination orderings, spoiler effects, and ``monotonicity" paradoxes wherein voters switching their preference to a given candidate can sometimes ensure that candidate loses the election \citep{arrow2012social, PaG3955,graham-squire,drutman2021we, McCuneWilson2023Spoiler}. Closest to our work are algorithmic approaches for computing the minimum number of ballot perturbations required to change the outcome of an election, as well as the overall sensitivity of RCV results to ballot modifications \citep{cary2011estimating, blom2019toward, blom2020did, deshpande2026simpler}.

The methodological literature on uncertainty quantification for ranked choice voting polls is sparse. Moreover, few published polls include justifications for their reported uncertainty measures. For this reason, we conducted a survey across polls for both the 2021 and 2025 New York City Democratic primaries -- the former of which led to the selection of Eric Adams as the Democratic nominee, and the latter of which led to the ascension of Zohran Mamdani. We considered the uncertainty measures in polls conducted by twelve prominent pollsters \citep{dfp2021nyccrosstabs, noauthor_citizen_nodate,noauthor_marist_nodate,noauthor_ipsos_2021,noauthor_emerson_nodate,services_change_2021,bridges_siena_2025,manhattan_institute_manhattan_nodate,csp_center_2025,mumford_emerson_2025,empire_empire_nodate,singh_yale_2025,marist_marist_nodate}. 

In each poll, respondents were asked to rank their preferences. The polls reported the resultant rankings, and most pollsters also simulated rounds of ranked choice voting until a winner was identified. The findings are summarized in Table \ref{tab:polls}. Under ``Uncertainty," we report ``MoE" when margins of error are used, and ``mMoE" when ``modeled" margins of error are used, which incorporate the impact of weighting individual survey responses (for more details, see the Appendix, Section \ref{sec:mmoe}). 

\begin{table}[h]
\centering
\resizebox{\textwidth}{!}{
\begin{tabular}{lllrrcc}
\toprule
Year & Pollster & Uncertainty & \makecell{Sample\\Size $n$} & MoE (\%) &
\makecell{Different\\MoEs by Round} &
\makecell{Alternative\\Uncertainty} \\
\midrule
2021 & Manhattan & MoE & 500  & 4.4 &  &  \\
2021 & Change & mMoE & 822  & 3.6  &  &  \\
2021 & FairVote & MoE & 800  & 3.5  &  & Bootstrap \\
2021 & Emerson & MoE & 1,284 & 2.7  &  &  \\
2021 & Marist & MoE & 876  & 3.8  & \checkmark &  \\
2021 & Ipsos & mMoE & 906  & 5.7  &  &  \\
2021 & Data for Progress & MoE & 1354 & 3.0 &  &  \\
\midrule
2025 & Siena & mMoE & 556  & 4.9  &  &  \\
2025 & Emerson & MoE & 800  & 3.4  & \checkmark &  \\
2025 & YouGov & mMoE & 416  & 6.7  &  &  \\
2025 & Empire & MoE & 3,012 & 1.8  &  &  \\
2025 & Manhattan & MoE & 1,000 & 3.1 &  &  \\
2025 & Marist & mMoE & 1,350 & 4.3 &  &  \\
2025 & Center for Strategic Politics & MoE & 580 & 4.1 &  &  \\
\bottomrule
\end{tabular}
}
\caption{Poll uncertainty measures by year and pollster.}
\label{tab:polls}
\end{table}

We verified that all reported MoEs were conservative by deriving the reported MoEs using reported sample sizes. Five pollsters incorporated the design effect to compute mMoEs. Two pollsters reported conservative round-level -- as opposed to survey-level -- MoEs. Conservative MoEs become wider in later rounds because the sample size (i.e. the number of active ballots) declines. While most pollsters did not publish a methodology, four pollsters stated that their MoE estimates incorporated the design effect from weighting, and six pollsters reported MoEs equal to the conservative MoE estimate for the given sample size. Emerson College's poll for NYC 2025 gave separate MoEs for the first and last round of polling, both conservative estimates based on the number of ballots in each round, while all other polls reported only a single, survey-level MoE. Among alternative uncertainty measures, only FairVote ran a bootstrap simulation based on their poll results for NYC 2021. The survey shows that the conservative survey-level MoE has been adopted as the standard method to communicate sampling uncertainty in RCV elections. 

\subsection{Election Datasets}

To demonstrate practical challenges in deploying standard methods on ranked choice voting polls, we collected datasets from two recent, high-profile RCV elections. In each case, we obtained the cast vote records from the relevant board of elections. Analyses of cast vote records have become an increasingly popular tool for understanding voter preferences \citep{kuriwaki2024cast}. They are particularly useful in the context of ranked choice elections, because random sampling from the CVR allows us to simulate unbiased polls, in which voters' preferences are truly selected from the population with equal probability \citep{shirani2018disentangling}. 

\subsubsection{New York City 2021 Democratic Primary}

The first election we consider is the June 22, 2021 Democratic mayoral primary in New York City. This was the first citywide election conducted using ranked choice voting, after a referendum for RCV adoption was passed in 2019. Voters were permitted to rank up to five different candidates in each race. 

The mayoral primary included 13 official candidates, and just shy of one million ballots were cast. The frontrunners were Brooklyn Borough President Eric Adams, City Sanitation Commissioner Kathryn Garcia, civil rights lawyer Maya Wiley, and venture capitalist Andrew Yang. Many pre-election polls projected that Adams would win comfortably, but he ultimately triumphed in the final instant runoff round by a margin of only 0.8\%. This was due to the fact that his rival in the final round was -- somewhat unexpectedly -- Garcia, rather than Wiley, whom Adams would have defeated more comfortably. Adams went on to win the general election and serve for one term as the city's mayor.

Along with the cast vote records, we also obtained individual-level poll responses from the final poll conducted by Data for Progress, an American think tank and polling firm \citep{dfp2021nyccrosstabs}. The poll was conducted from June 18 to June 20, and asked likely Democratic primary voters citywide for their preferences for the mayor and comptroller races. Respondents in Manhattan were also asked about their preferences for the District Attorney and Borough President races. The citywide races had n = 1,354 respondents with a reported conservative survey-level margin of error (MoE) of plus or minus three percentage points. These responses were weighted by age, gender, education, race, borough, and voting history to reflect the likely citywide electorate.

%The poll was conducted from June 18 to June 20, and asked likely Democratic primary voters for their preferences for the mayor and comptroller races. The poll for mayor had $n = 1,354$ respondents with a reported conservative survey-level margin of error (MoE) of $\pm 3\%$. Poll responses were weighted by age, gender, education, race, borough, and voting history to reflect the likely electorate. 

%The poll concluded that Adams, Garcia, and Wiley were likely the final three candidates. In the penultimate round, Adams led the poll ($38\%$ support) and Garcia and Wiley were nearly tied ($31\%$ support each). The pollsters ran scenarios eliminating either Wiley or Garcia. When Garcia was eliminated, Adams won the final round against Wiley at a margin of $54\%$ to $46\%$. When Wiley was eliminated, Garcia clinched the victory over Adams with $52\%$ to $48\%$ of the vote. Given the tie between Wiley and Garcia in the penultimate round and a statistically insignificant final-round matchup between Adams and Garcia, the pollsters concluded that the race was ``firmly too close to call when considering the full range of possibilities that small changes in voter rankings could cause" \citep{dfp2021nyccrosstabs}. 

The poll concluded that Adams, Garcia, and Wiley were the top contenders in the race. In a memo published the day before the election, the pollsters simulated rounds of ranked choice voting \citep{dfp2021nycrcvmemo}. In the penultimate, simulated round (for which $n = 1,220$ respondents' ballots remained active), the pollsters found Adams led with $38\%$ support, while Garcia and Wiley were tied with $31\%$ support each. Given the tie, the memo considered scenarios in which Wiley or Garcia were eliminated first. When Garcia was eliminated first, Adams won the final round against Wiley by a margin of 54\% to 46\%. However, when Wiley was eliminated first, Garcia led Adams 52\% to 48\% in the final round, a lead well within the margin of error. Given the tie between Wiley and Garcia in the penultimate round and a statistically insignificant final-round matchup between Adams and Garcia, the memo concluded that the race was ``firmly too close to call when considering the full range of possibilities that small changes in voter rankings could cause" \citep{dfp2021nycrcvmemo}.

\subsubsection{Alaska 2022 U.S. House Special}

We also obtained cast vote records from the August 16, 2022 U.S. House special election in Alaska. The election was triggered by the death of longtime incumbent Don Young, and drew about 190,000 voters. It was the first election in Alaska to use ranked choice voting, after the system had been approved by voters in a 2020 ballot referendum. 

Three candidates -- former State Representative Mary Peltola, former governor Sarah Palin, and Nick Begich III, scion of a powerful Alaskan political family -- competed in the election. Peltola defeated Palin in the final round by about three points, becoming the first Democrat elected statewide in Alaska in nearly fifteen years. However, her victory was contingent on the fact that Begich had been eliminated in the first round. Based on the underlying ballot preferences, Begich actually would have defeated Peltola in a one-on-one matchup \citep{graham2022mathematical}.

This reflects a well-known property of ranked choice voting. A candidate who is preferred one-to-one against each of their competitors is known in the parlance of social choice theory as a ``Condorcet winner" \citep{de2014essai, arrow2012social} -- but RCV elections are \emph{not} guaranteed to select the Condorcet winner. In this election, Begich was a Condorcet winner: the cast vote records show that a majority of voters preferred him to Palin, and a majority of voters preferred him to Peltola. However, because he drew the fewest first-place ballots, he was eliminated before he could compete head-to-head against either candidate.

\section{Unique Challenges in RCV Polls} \label{sec:challenges}

\subsection{Sampling Uncertainty Need Not Favor True Winners}\label{sec:samplingBias}

\begin{figure}[t]
    \centering
    \includegraphics[width=0.60\textwidth]{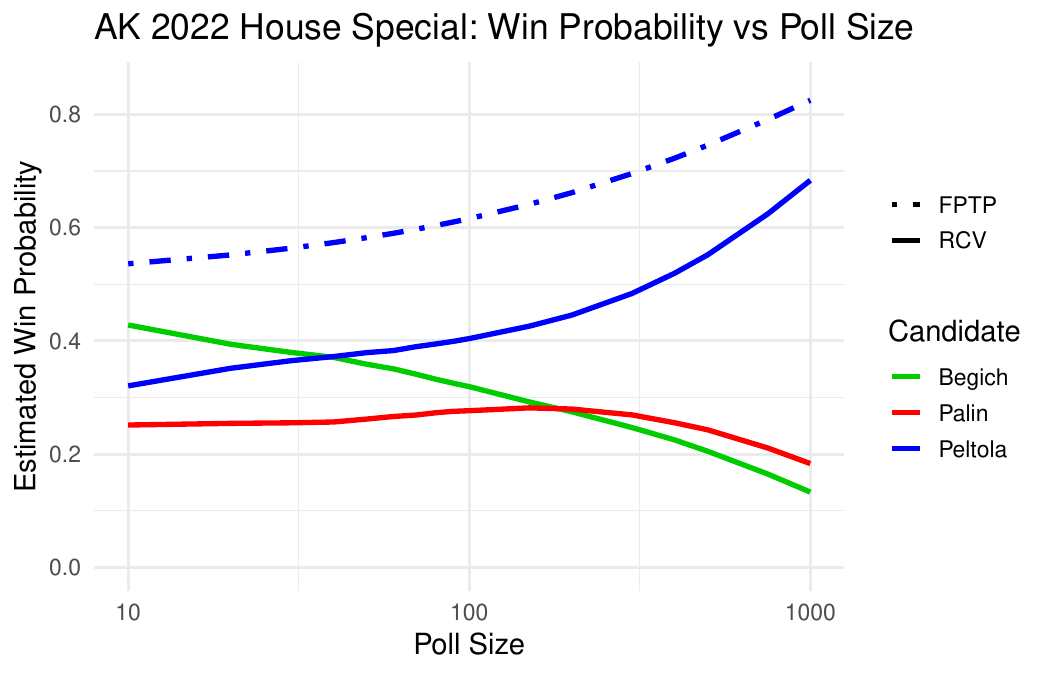}
    \includegraphics[width=0.39\textwidth]{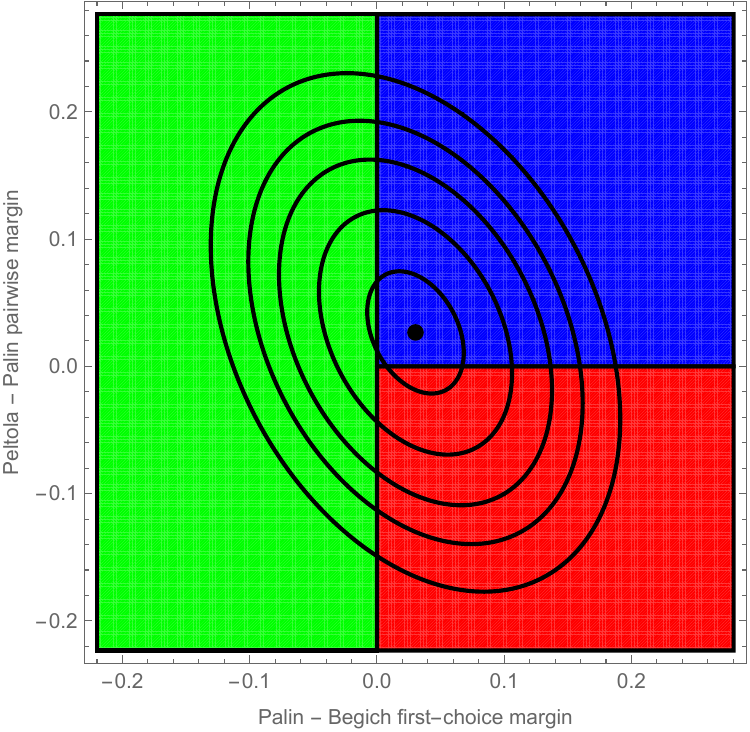}
    \caption{The left panel plots the probability that each candidate wins, in an unbiased sample from the cast vote records of the 2022 Alaska House special election. Solid lines show RCV win probabilities, estimated from 500{,}000 random samples of each size. The dotted blue line shows a two-candidate first-past-the-post (FPTP) benchmark with the same final-round Peltola-Palin margin of 2.96 percentage points. The right panel shows the geometry underlying these probabilities. The x axis gives the sampled Palin-Begich first-choice margin, while the y axis gives the sampled Peltola-Palin pairwise margin. Colored regions correspond to the winner implied by these margins (under the assumption that remaining election margins retain their true signs). The ellipses are contours of the joint sampling distribution centered at the true election outcome. }
    \label{fig:ak-2022-rcv-fptp}
\end{figure}

We begin with an empirical example. In the left panel of Figure \ref{fig:ak-2022-rcv-fptp}, we repeatedly draw 500,000 samples of different sizes (ranging from 10 to 1,000) from the cast vote records of the 2022 U.S. House special election in Alaska. Each sample can be considered an unbiased poll, as the individual voters have identical probabilities of being selected.  Hence, we are considering the effect of \emph{only} sampling variability, not selection bias or other forms of survey error. 

Within each unbiased sample, we compute the winner under the standard ranked choice voting algorithm. We then estimate the probability that each candidate wins in a poll of the given size by calculating the frequency of victory across the half-million pseudo-polls. The left panel of Figure \ref{fig:ak-2022-rcv-fptp} plots these estimated win probabilities as a function of the poll size. The blue, dotted line provides a reference comparison based on a hypothetical, two-candidate FPTP election between Peltola and Palin. It gives the probability that Peltola would lead in an unbiased poll if the election were a head-to-head contest with the same final margin observed in the actual election, Peltola +2.96\%.

The left panel of Figure \ref{fig:ak-2022-rcv-fptp} reflects several curious patterns that can emerge in unbiased ranked choice voting polls. At sufficiently small poll sizes, the most frequent winner in the pseudo-polls is Begich, even though he was actually eliminated first in the election. Such behavior is impossible in unbiased polls of  FPTP elections: at any sample size, the majority of the mass of the sampling distribution must lie within the region in which the true winner is the victor. Here, sample sizes of approximately 50 respondents are required for Peltola to emerge as the most probable victor. A second pathology in Figure  \ref{fig:ak-2022-rcv-fptp} is the evolution of the win probabilities for Palin, which do not decrease monotonically with the sample size. Palin actually wins in about 28\% of unbiased polls at sample sizes of 150 or 200 respondents, vs. about 25\% of unbiased polls at sample sizes of 10 or 20 respondents. Again, non-winning candidates cannot improve their win probabilities at larger sample sizes in FPTP polls. 

These behaviors emerge from two key features of the underlying ballot preferences. First, in the true election, Begich was eliminated by a roughly 3\% margin: he received 28\% of first-choice votes in the first round, vs. 31\% for Palin. In the second round, Palin also fell to Peltola by approximately 3\%. These narrow margins are difficult to detect with even a moderately large sample size; the reference line atop the plot shows that only about 83\% of polls at size $n = 1,000$ would correctly detect Peltola's final-round victory margin of 3\% in a FPTP race. Second, the elimination order matters. Recall that Begich would have defeated Peltola had he advanced to the second round. Many small-sample polls will mistakenly estimate that Palin is eliminated first, and the majority of these polls will estimate that Begich is the overall victor, as he truly was preferred to Peltola by a majority of voters.

The right panel of Figure \ref{fig:ak-2022-rcv-fptp} helps visualize the fragility of the RCV winner, especially in small polls. On the x axis, we plot the margin between Palin and Begich among first choice ballots. On the y axis, we plot the Peltola vs. Palin ``pairwise" margin -- that is, among all voters, the frequency of those who rank Peltola above Palin minus the frequency of those who rank Palin above Peltola. These are the two margins that were approximately 3\% in the election, and the true value is represented by the black dot in the upper right. 

The x and y axes represent the narrowest margins in the actual election. Peltola's true first-choice margin over Palin was nearly 9\%, while Begich's pairwise margin against Peltola was over 5\%. Hence, we make the simplifying assumption that any polls exhibiting a negative Palin-Begich first-choice margin will show Begich as the eventual winner (since Palin will be eliminated first, and Begich will defeat Peltola head-to-head). In the case of a positive Palin-Begich first-choice margin, we assume the poll's winner is decided by the Peltola-Palin pairwise margin. Under these assumptions, the winner regions are colored appropriately: green for Begich, red for Palin, and blue for Peltola. 

The sampling distribution for these margins will be approximately jointly normal via the Central Limit Theorem. It will be centered at the true parameter values, with a covariance matrix computable from the true population parameters. Level sets for this distribution are superimposed on top of the plot, with larger ellipses corresponding to the sampling distribution for small polls, and smaller ones corresponding to larger polls. Geometrically, we can see the consequences of the fact that the true parameter values are quite near to the decision boundary. Among small polls, the plurality of the mass lies in the green region, corresponding to a Begich victory. Only as the poll size grows -- such that the sampling distribution contracts around the true parameter value -- does the majority of the mass lie in the blue region corresponding to a Peltola victory. 

In a two-candidate FPTP election, the winner is determined by a single margin, so the parameter space is divided by a single decision boundary. As the sampling distribution shrinks around the true parameter value, the probability of identifying the true winner necessarily increases. RCV elections are governed by a different geometry. The winner depends on multiple margins simultaneously, producing a partition of the parameter space into distinct winner regions, and small perturbations in the sampled margins can move a poll across multiple decision boundaries.

More broadly, Figure \ref{fig:ak-2022-rcv-fptp} illustrates that point estimation is especially fraught in RCV elections. If the election is determined by narrow margins and sample sizes are not particularly large, samples can easily cross one or multiple decision boundaries. Uncertainty quantification is even more important in such cases. 

\subsection{Standard Methods Break Down in RCV Polls}\label{sec:standardMethods}

Per our discussion in Section \ref{sec:currPractice}, the standard practice in ranked choice voting polls is to simulate instant runoff elections based on the responses $y_i, i = 1, \dots, n$ given in the poll. Most commonly, a single, conservative margin of error is reported for the entire poll, though sometimes different MoEs are reported for different elimination rounds. 

In Table \ref{tab:elims_below_moe}, we sample repeatedly from the cast vote records of both the Alaska and NYC elections and compute the average proportion of elimination margins that fall below a conservative margin of error computed at the $\alpha = 0.05$ level. We draw $50,000$ samples at each sample size $n$, and consider a range of values corresponding to modestly-sized and large polls. While there are only three candidates in the Alaska special, there are 13 in the NYC mayoral primary. We can see that -- at least in heavily contested elections such as these -- statistically insignificant elimination margin estimates are the norm, rather than the exception. The problem could be partially mitigated by using unbiased estimates for the round-specific elimination MoEs. But the fundamental problem remains: testing many small margins simultaneously means that some are likely to be indistinguishable from sampling variability. 

\begin{table}[ht]
\centering
\label{tab:elims_below_moe}
\begin{tabular}{rcc}
\toprule
Sample Size $n$ & NYC 2021 Mayoral Primary & Alaska 2022 Special Election \\
\midrule
600  & 91\% & 95\% \\
800  & 87\% & 94\% \\
1,000 & 84\% & 92\% \\
1,200 & 83\% & 90\% \\
1,400 & 82\% & 89\% \\
\bottomrule
\end{tabular}
\caption{Proportion of elimination margins falling below the conservative margin of error (MoE) under repeated sampling from cast vote records. Results are shown for the 2022 Alaska U.S. House Special Election (3 candidates) and the 2021 New York City Democratic Mayoral Primary (13 candidates).}
\end{table}

This yields a tension in interpretation. When there are a large number of candidates, small true elimination margins, or modestly-sized polls, we expect some estimated elimination margins to lie within the margin of error. In Section \ref{sec:samplingBias}, we saw that elimination orders did affect the true winner in the Alaska House special. But in many cases, early elimination orders are irrelevant to the outcome of the election. For a more typical case, see Figure \ref{plot:elimOrders}. Here, we sample 50,000 times from the cast vote record of the  2021 New York Democratic mayoral primary. Our samples are of size $1,354$, matching the size of the Data for Progress poll. Within each sample, we simulate the instant runoff algorithm on the sampled ballots, and record the order in which candidates are eliminated, such that the first-eliminated candidate is assigned a ``1," the second-eliminated candidate a ``2," etc. 

\begin{figure}
    \centering
    \includegraphics[width=1\linewidth]{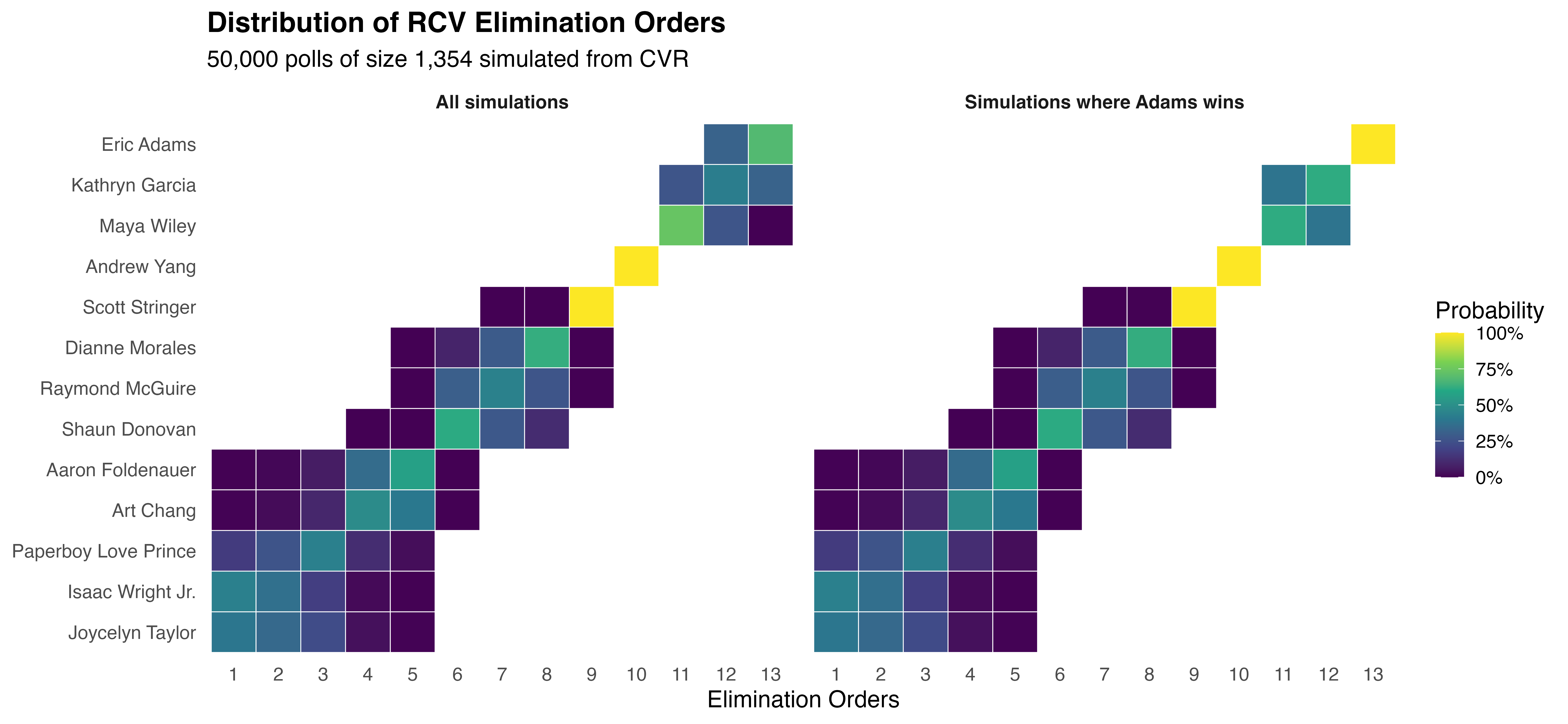}
    \caption{Order in which each candidate is eliminated in simulated instant runoffs in random samples of size $1,354$ drawn from the cast vote record of the 2021 New York Democratic mayoral primary. Left panel: elimination orders among all 50,000 simulations. Right panel: elimination orders in the 68\% of simulations in which Eric Adams wins.}
    \label{plot:elimOrders}
\end{figure}

In the left panel of Figure \ref{plot:elimOrders}, we plot the distribution of elimination orders by candidate across all 50,000 simulations. In the right panel, we restrict to the 68\% of samples in which Eric Adams is the winner. Within each plot, we note that there is wide variability in the elimination orders for low-polling candidates; for example, we see that Aaron Foldenauer and Art Chang are eliminated first, second, third, fourth, fifth, and sixth with positive probability across the samples. However, across the two plots, the distributions are almost identical for all candidates except Eric Adams, Kathryn Garcia, and Maya Wiley. This is due to the fact that Andrew Yang had notably higher support than all other candidates except these three candidates, but notably lower support than any of these three candidates. Hence, Yang serves as a ``bottleneck'' -- candidate eliminations prior to Yang's elimination can occur in any order without affecting the ultimate RCV winner. As a result, considerable uncertainty about exact elimination margins and orderings can coexist with stable probabilities of the ultimate winner. 

Herein lies the central difficulty. Pollsters are typically interested in identifying the election winner, and assessing whether their lead is statistically fragile. However, they are often unable to detect the precise elimination order with high statistical confidence. And while many elimination orders do not alter the electoral winner, some do. In cases where the path dependency of the election is well understood, pollsters could report a small number of alternative scenarios corresponding to different elimination pathways -- e.g., reporting both the Adams vs. Wiley and Adams vs. Garcia final round margins in the NYC mayoral primary. However, this approach relies on substantial understanding of the race dynamics and becomes infeasible if there are many determinative eliminations. 

Fundamentally, these challenges reveal a mismatch between the uncertainty procedures -- which quantify uncertainty in individual elimination margins (and therefore in particular elimination orders) -- and the objective, which is to assess uncertainty about the election winner.

\subsection{Frequentist Approaches to RCV Poll Uncertainty}

As discussed in the prior section, reporting point estimates along with a margin of error is insufficient to quantify uncertainty in RCV polls. However, the problem can be formulated in terms of frequentist hypothesis testing procedures. In this section, we formalize the problem but showcase some of the drawbacks to a purely frequentist approach. 

An inherent aspect of winner identification in ranked choice voting is that it is a high-dimensional problem, as demonstrated by the following result. 

\begin{lemma}\label{lemma:numParams}
Suppose we have an RCV election with $K$ candidates, and we suppose that no winner is declared until the final round. Identifying the winner is a function of 
\[  \sum_{s=2}^K {K \choose s}(s-1) = 2^{K-1}(K-2)+1 \]
parameters, where each is a ``margin" between two candidates in a given elimination round. 
\end{lemma}

\begin{proof}
See the Appendix, Section \ref{proof:paramsLemma}. 
\end{proof}

Lemma \ref{lemma:numParams} means that five margins are needed to identify the winner of an RCV election with three candidates; 17 margins are needed with four candidates; and 49 are needed with five candidates. This is not in-and-of-itself problematic. Sample estimates for each of these margins are averages of independent and identically distributed measurements, and hence obey a Central Limit Theorem. Thus, we can estimate the covariance matrix and construct a Wald confidence ellipsoid that asymptotically attains a desired frequentist coverage rate \citep[see e.g.][]{van2000asymptotic}. 

In practice, however, this approach is poorly aligned to the inferential target of identifying the winner. The regions in which each of the $K$ candidates win are themselves unions of polyhedra, which together partition the space. In contested elections, much of the confidence ellipsoid's volume is expended to estimate parameters that are far from decision boundaries, but this makes the confidence set too wide in the directions of close margins, and they frequently cross decision boundaries. 

As a simple demonstration of this pathology, see Table \ref{tab:ak_wald_coverage}. We sample from the CVR of the Alaska House special, this time randomly drawing 10,000 pseudo-polls of various sizes. For each pseudo-poll, we compute a 95\% Wald confidence set and record its coverage rate for the five decision margins. In the second column, we report what proportion of these confidence ellipsoids overlap with all \emph{three} winner regions. In the third column, we record what proportion lie solely with the region corresponding to the true winner, Peltola. The final column reports the simultaneous coverage rate of the five margin parameters. Notably, while we achieve close-to-nominal simultaneous coverage of the parameters, winner identification is extremely poor: even at a sample size of $1,400$, almost 97\% of confidence ellipsoids overlap with all three winner regions. 

\begin{table}[ht]
\centering
\begin{tabular}{rrrrr}
\toprule
Sample size $n$
& All winner regions
& True winner only
& Other subsets
& Simultaneous coverage \\
\midrule
600  & 99.0\% & 0.0\% & 1.0\% & 94.3\% \\
800  & 98.5\% & 0.0\% & 1.5\% & 94.6\% \\
1000 & 97.7\% & 0.0\% & 2.3\% & 94.4\% \\
1,200 & 97.3\% & 0.0\% & 2.7\% & 94.8\% \\
1,400 & 96.8\% & 0.01\% & 3.2\% & 94.6\% \\
\bottomrule
\end{tabular}
\caption{Performance of the five-dimensional Wald confidence ellipsoid in 10,000 simulated polls, sampled randomly from the cast vote records of the the 2022 Alaska House special election. The confidence set approximately attains the nominal 95\% simultaneous coverage rate for the underlying parameter vector at all sample sizes, but it almost never identifies a unique winner. Even at $n=1,400$, the ellipsoid intersects all three winner regions in 97\% of simulations, and isolates the true winner alone in only 0.01\% of simulations.}
\label{tab:ak_wald_coverage}
\end{table}

Certainly, alternative frequentist testing procedures -- such as generalized likelihood ratio tests or intersection union tests \citep{silvapulle2011constrained} -- may exhibit higher power for the detection of winners in RCV polls. However, two fundamental issues make the problem challenging for frequentist testing procedures. First, the dimension of the problem grows rapidly with the number of candidates, even as the underlying decision margins may remain narrow. Second, the underlying geometry of the problem is formidable: winner regions are not nested, and comprise complex unions of polyhedral regions. 

\section{Proposed Bayesian Uncertainty Quantification} \label{sec:method}

Given the challenges identified in Section \ref{sec:challenges}, we advocate a Bayesian approach that focuses directly on the quantity of interest: the probability that each candidate wins the election. Rather than conducting inference on somewhat abstruse margins in theoretical elimination rounds, we model uncertainty in the distribution of ballot rankings and ask how often each candidate emerges as the winner. The resulting posterior win probabilities naturally incorporate sampling error, path dependency, and other features of RCV elections.

\subsection{Preliminaries}

We now reformulate the RCV problem in terms of the ballot-ordering distribution $\boldsymbol{\pi}$. Recall, from Section \ref{sec:assumptions}, the responses $y_i$ of voters $i = 1, \dots, N$ in the population will take values in $\mathcal{P}$, the finite set of all possible ballot orderings among the set of candidates $\mathcal{C}$. Moreover, we define $\boldsymbol{\pi} = (\pi_1,\ldots,\pi_{|\mathcal P|})$ to be the population-level probability that each voter submits a ballot of each type, i.e.
\[ \pi_j = \Pr(y_i = p_j). \]
Lemma \ref{lemma:numParams} gives the minimal set of parameters that are necessary to determine the RCV winner. However, observe that the entries of $\boldsymbol{\pi}$ are also sufficient to determine the winner -- this is simply a less parsimonious representation of the data. Moreover, observe that for units $i = 1, \dots, n$ in a survey, we can compute the sample analogue $\boldsymbol{\hat \pi} = (\hat \pi_1,\ldots,\hat \pi_{|\mathcal P|})$ where
\[ n_j = \sum_{i = 1}^n \mathbbm{1}(y_i = p_j) \qquad \text{and} \qquad \hat \pi_j = \frac{n_j}{n}. \]

Because the survey samples $y_i, i = 1, \dots, n$ are assumed independent and identically distributed, it follows immediately that the count vector $(n_1,\ldots,n_{|\mathcal P|})$ follows a multinomial distribution, with parameters
\[ (n_1,\ldots,n_{|\mathcal P|}) \sim \text{Multinomial}(n, \boldsymbol \pi). \]
Lastly, we recall that, under this model, $(n_1,\ldots,n_{|\mathcal P|})$ is a minimal sufficient statistic for $\boldsymbol \pi$. Because the RCV winner is a deterministic function of $\boldsymbol \pi$, it follows immediately that this statistic contains all information in the sample relevant to determining the winner.

\subsection{Bayesian Inference on RCV Polls}\label{sec:coreMethod}

In this section, we assume familiarity with standard Bayesian concepts such as a prior distribution, likelihood function, posterior distribution, and conjugacy. For a brief review, see the Appendix, Section \ref{sec:bayes}; for a deeper discussion, see e.g.  \cite{bolstad_introduction_2007}. 

Our goal via Bayesian inference is to estimate the posterior probability that each candidate wins the election, conditional on the survey data. The posterior distribution can be sampled to obtain a simulated electorate on which an RCV winner can be computed. Over repeated sampling from the posterior, we can obtain proportions of election victories for each candidate. These proportions can be interpreted as each candidate's conditional probability of winning. 

To obtain the posterior, we make use of a well-known conjugacy relationship: that of the Dirichlet and the multinomial distribution. Per the discussion in the prior section, we know that the count vector obeys a multinomial distribution, with associated density
\[ f(n_1\ldots, n_{|\mathcal{P}|} \mid \pi_1, \ldots, \pi_{|\mathcal{P}|} ) = \frac{n!}{\prod_{j = 1}^{|\mathcal{P}|}n_j!}\prod_{j = 1}^{|\mathcal{P}|} \pi_j^{n_j}. \]

The Dirichlet distribution -- a generalization of the Beta distribution -- is a conjugate prior for a multinomial likelihood. If we assume the prior
\begin{align*}
\boldsymbol{\pi} \sim \mathrm{Dir}(\boldsymbol{\alpha}) \qquad \text{where} \qquad \boldsymbol{\alpha} =
(\alpha_1,\ldots,\alpha_{|\mathcal P|}),
\end{align*}
then the density of $\boldsymbol{\pi}$ is
\begin{align*}
g(\pi_1,\ldots,\pi_{|\mathcal P|}; \boldsymbol{\alpha}) = \frac{1}{B(\boldsymbol{\alpha})} \prod_{j=1}^{|\mathcal P|} \pi_j^{\alpha_j-1}, \qquad \text{where}  \qquad  B(\boldsymbol{\alpha}) = \frac{\prod_{j=1}^{|\mathcal P|}\Gamma(\alpha_j)} {\Gamma\!\left(\sum_{j=1}^{|\mathcal P|}\alpha_j\right)}
\end{align*}
is the multivariate Beta function. If we suppose $(n_1,\ldots,n_{|\mathcal P|}) \sim \text{Multinomial}(n, \boldsymbol \pi)$, then conjugacy gives us 
\[ \boldsymbol{\pi}  \mid (n_1,\ldots,n_{|\mathcal P|}) \sim \mathrm{Dir}\left(\alpha_1 + n_1, \dots, \alpha_{|\mathcal{P}|} + n_{|\mathcal{P}|}\right).\] 

Using this relationship, we can obtain a straightforward posterior distribution on the ballot type probabilities, without having to resort to computational methods. This method requires a prior distribution; we suggest as a starting place the flat, uninformative prior $\boldsymbol \pi \sim \text{Dir}(1, \dots, 1)$, indicating that all ballot types are equally likely. However, preexisting polling data can replace the flat prior if available. 

A posterior distribution on $\boldsymbol \pi$ does not immediately translate to candidate victory probabilities. However, we can sample repeatedly from the posterior distribution to obtain samples $\boldsymbol{\tilde \pi^{(1)}}, \dots, \boldsymbol{\tilde \pi^{(B)}}$, for a large value of $B$. Each sampled value $\boldsymbol{\tilde \pi^{(\ell)}}$ is a simplex-restricted, $|\mathcal{P}|$-dimensional vector of proportions for each of the unique ballot rankings. Hence, each such value can be interpreted as an electorate simulated from the posterior, where the simulated electorate data has been expressed in terms of the proportions of ballot types. Crucially, the instant runoff algorithm can be run straightforwardly on any of the sampled $\boldsymbol{\tilde \pi^{(\ell)}}$ values, and each will identify a winning candidate. 

Hence, repeated sampling from the Dirichlet posterior yields a distribution of victories by candidate. This simulates the effect of random variation in the distribution of ballots on election outcomes. The proportion of simulated electorates in which a candidate wins the election is interpreted to be the probability that the candidate wins the election according to the poll.

\subsection{Interpretation of the Posterior Distribution and Estimated Probabilities}

The posterior win probabilities constitute the primary uncertainty measure produced by this method. Unlike margins of error or confidence regions on individual elimination margins, these probabilities directly quantify uncertainty about the identity of the election winner. This distinction is particularly important in RCV elections, in which uncertainty about elimination orders often co-occurs with limited uncertainty about the ultimate election winner. 

We advocate reporting the estimated win probabilities alongside the point estimated winner from the poll. For example, a candidate with an estimated victory probability exceeding 95\% can reasonably be viewed as a highly likely winner, while candidates with victory probabilities below 5\% can be ruled out as plausible winners. More generally, however, the probabilities provide a continuous measure of uncertainty rather than forcing pollsters or readers into binary conclusions. A candidate with a 70\% estimated probability of victory is neither certain to win nor involved in a completely uncertain race. By reporting the full distribution of victory probabilities, pollsters can communicate uncertainty in an interpretable manner.

\subsection{Reducing the Set of Possible Rankings} \label{sec:pruning_section}
One practical challenge faced when simulating an RCV electorate is the large size of the sample universe. Given a pool of $K$ candidates and allowing voters to rank up to $R$ candidates so that $R \leq K$, the number of unique potential rankings (UPR) is given by 
\[ \sum_{r=1}^{R} \frac{K!}{(K-r)!}. \]

Thus, as the number of candidates increases, the number of unique rankings increases super-exponentially. This does not pose a problem in some jurisdictions, such as Alaska, where only the top four vote-getters in a nonpartisan primary move on to a ranked choice voting general election. But for races with a large number of candidates, a typical poll will be a very sparse sample of the set $\mathcal{P}$ of unique ballot rankings. 

For example, in the 2021 NYC Democratic mayoral primary, there were 13 candidates and voters were allowed to rank up to five, so $|\mathcal{P}| = 173,485$. Only about 43\% of these orderings -- 75,153 to be precise -- appeared in the cast vote records. However, the Data for Progress poll, which had a relatively large sample size of $n = 1,354$, included only 894 unique ballot rankings. Because the posterior distribution is formed based on the observed ballot frequencies in the poll, this can pose a problem for inference. 

One solution is to include all possible UPR in the prior and implement a variation of additive smoothing \citep{Jeffreys_1946} by setting the prior probabilities for the unobserved entries of $\mathcal{P}$ to be half the size of the regular prior probability. The posterior distribution will assign a low, but nonzero, probability to unseen ballots. However, a major drawback of this approach, visible from the cast vote records in New York, is that many possible ballots do not appear at all in the true electorate's preferences. Therefore, by including all unique possible rankings, we incorrectly include rankings that are not present in the electorate. 

An alternative approach involves first recognizing that many of the entries in $\mathcal{P}$ can be ignored or canonicalized. This is closely related to our observation in Section \ref{sec:standardMethods} that many kinds of variability in ballot preferences do not affect the eventual winner of the election. We formalize this idea via the following lemma. 

\begin{lemma}[Ballot Pruning]\label{pruneThm}
Suppose we have an RCV election involving candidate set $\mathcal{C}$, and voters may rank up to $R \leq K = |\mathcal{C}|$ candidates. Suppose also that an oracle provides the true final set of $\mathcal{C}^\star_m$ candidates, where $|\mathcal{C}^\star_m| = m$ and the entries in $\mathcal{C}_m^\star$ are precisely those candidates that remain after the first $K - m + 1$ elimination rounds. If each ballot is ``pruned" to replace candidates in $\mathcal{C} \setminus \mathcal{C}_m^\star$ with blanks, then RCV run on the pruned ballots will yield the same winner as RCV run on the unpruned set of ballot rankings. In fact, the round-by-round vote tallies for each candidate in the pruned RCV results will be identical to the vote tallies for the final $m-1$ rounds of the unpruned RCV results.  
\end{lemma}

\begin{proof}
See the Appendix, Section \ref{sec:pruneProof}. 
\end{proof}

The ballot pruning lemma offers a way to reduce the number of relevant ballot rankings. It says that the order that minor candidates are eliminated does not matter if we know which candidates will remain. In other words, if we know who the final $m$ candidates are, we can collapse the ballots to ignore the other candidates; only the relative rankings of the major candidates affect the election outcomes. By considering only the candidates relevant to the election outcome, we can ensure that the sample is representative of the electorate's preferences on these major candidates.

The challenge lies in identifying the appropriate pruning set. If pollsters use an incorrect set of the final $m$ candidates, they may inadvertently distort the poll results, even if the omitted candidates are not actually competitive. Fortunately, our simulations suggest that this concern is limited in practice. Using the CVR data from the 2021 NYC Democratic mayoral primary, we drew 500 simulated polls at each sample size between $n = 100$ and $n = 1,000$, and identified the final three, four, and five candidates remaining under RCV tabulation in each sample. 

We recall from Figure \ref{plot:elimOrders} the unique dynamics of this race. The final five candidates remaining in the tabulations (in elimination order) were: Scott Stringer, Andrew Yang, Maya Wiley, Kathryn Garcia, and Eric Adams. At a sample size of $n = 1,354$, eliminations before Andrew Yang had no effect on the ultimate election winner. This dynamic manifests clearly in Figure \ref{plot:prune_cvr}, which reports the proportion of simulated polls for which the estimated set of final $m$ candidates matched the true set from the election, for $m \in \{3,4,5\}$. The horizontal line denotes a 99\% agreement threshold. The separation of Yang, Wiley,  Garcia, and Adams from the remainder of the candidates is evident by a sample size of $n = 400$, above which \emph{every} simulated poll identifies the top three and top four candidates correctly. For sample sizes of approximately $n = 600$ or greater, the final five candidates are also identified correctly in more than 99\% of simulations. This suggests that, in many practical settings, the uncertainty associated with pruning is substantially smaller than the uncertainty associated with the eventual election outcome itself.

\begin{figure}
    \centering
    \includegraphics[width=1\linewidth]{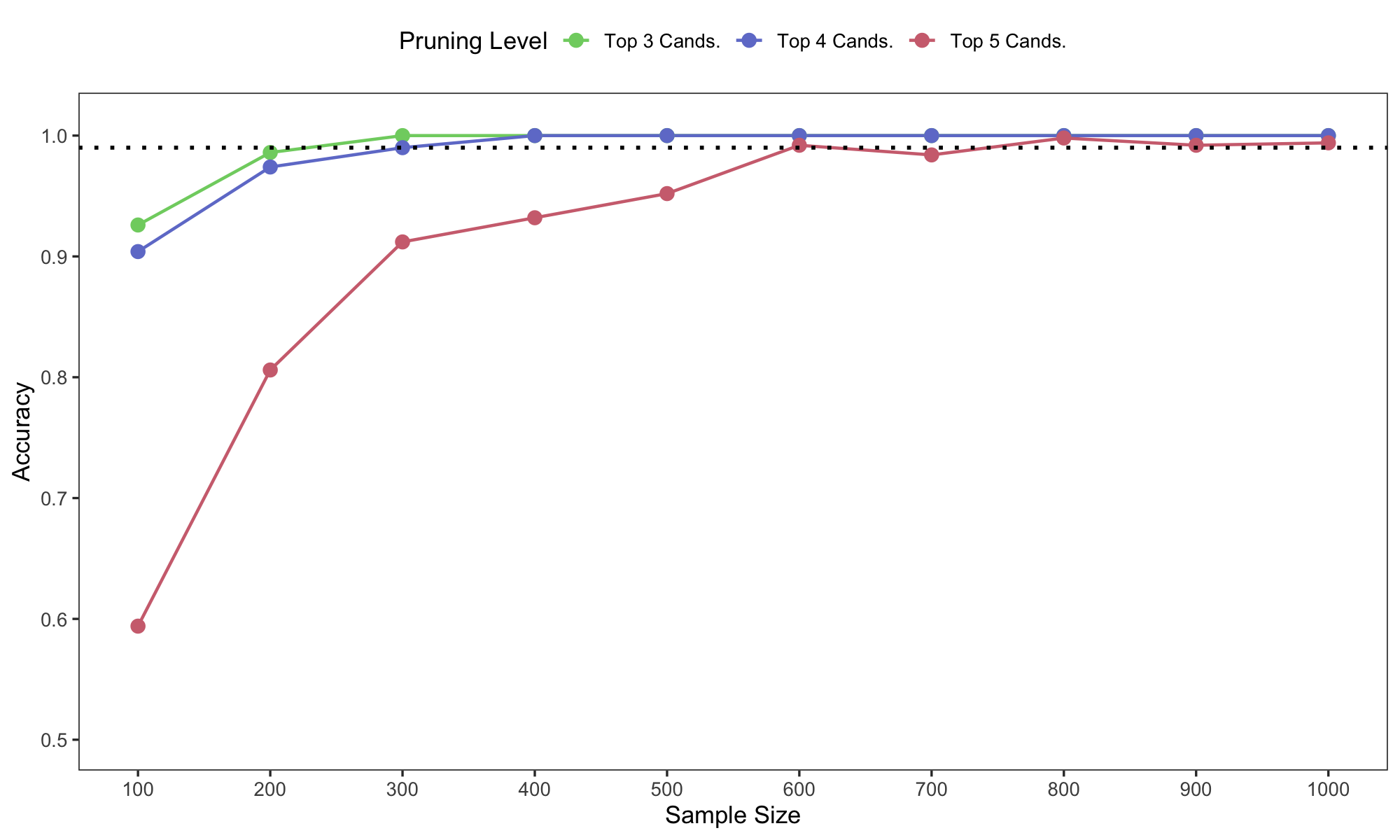}
\caption{Accuracy of candidate-set pruning under repeated sampling from the cast vote record of the 2021 New York City Democratic mayoral primary. The figure shows the probability that a poll correctly identifies the final $m$ candidates remaining under RCV tabulation. Even at moderate sample sizes, the candidates relevant to the election outcome are recovered with high reliability. The dashed horizontal line denotes a 99\% agreement threshold.}    
\label{plot:prune_cvr}
\end{figure}

In practice, of course, we cannot sample from the cast vote record prior to an election. However, we can bootstrap from the polling results to assess the stability of the pruning set. The following generalized method is recommended to decide the size of a pruning set with a high degree of confidence:

\begin{enumerate}
    \item Compute $B$ bootstrap samples on the poll, where $B$ is at least several hundred. 
    \item For each bootstrap, compute the order of eliminated candidates.% This is equivalent to finding a pruning set of every size.
    \item For each size of pruning sets, find the percent of the $B$ bootstrap samples that have identical pruning sets. 
    \item Choose the smallest sized pruning set with over a 99\% match rate across the samples $B$. 
\end{enumerate}

This process yields the smallest pruning set that is likely correct for both the poll results and the true election results. 

\subsection{Considerations about the Prior Distribution}\label{sec:poll_size}

Another consideration when employing Bayesian uncertainty quantification is the choice of prior. Per the discussion in Section \ref{sec:coreMethod}, a flat, uninformative prior over all the entries in $\mathcal{P}$ is often a reasonable starting place. 

However, using such a prior can inadvertently advantage candidates in cases where one candidate's supporters are more inclined to use expressive ballots. For example, consider the extreme hypothetical where one candidate's supporters always only rank that single candidate, while another candidate's supporters also rank other candidates on their ballots. A flat prior assigns ballot weights evenly to each unique ballot ranking. The first candidate will only appear on one unique ballot: the lone ballot ranking only him. The other candidate will appear on every combination of ballots ranking them and the other candidates. Even if a similar number of voters rank both candidates, the prior indirectly favors the latter candidate because the aggregate prior weight for ballots including the latter candidate is much larger than the single weight for the ballot ranking the former candidate. If the size of the prior distribution is relatively large compared to the size of the likelihood, the uninformative prior can dominate the posterior.

There are a few ways to mitigate this problem. One option is to downplay the effect of the prior by limiting the effective sample size associated with the prior. Under the conjugate update rules, where the posterior follows a $\text{Dir}(\alpha_1 + n_1, \ldots, \alpha_{|\mathcal{P}|} + n_{|\mathcal{P}|})$ distribution given a poll of sample size $n$, the effective sample size of the prior is given by
\[ n_{\text{eff}} = \sum_{j =1}^{|\mathcal{P}|} \alpha_j\]
One practical rule-of-thumb is to limit $n_{\text{eff}}$ to be no more than 15\% of the sample size \citep{bolstad_introduction_2007}. For a flat prior, this would simply mean 
\[ \alpha_j = \frac{0.15\,n}{|\mathcal P|} \qquad j=1,\ldots,|\mathcal P|.\] 

Moreover, when pruning yields a relatively small set of top candidates, the problem is effectively solved because the cardinality of $|\mathcal{P}|$ is dramatically reduced. For example, various ballots that rank only one major candidate and some other minor candidates will be collapsed to one unique ballot type after pruning. Reducing the number of unique ballot types that appear in the poll while keeping the sample size of the poll constant naturally limits $n_{\text{eff}}$ given a large enough sample size. The ideal solution is to first prune, then assess the effective sample size of the prior and, if necessary, reduce it to be no more than 15\% of the poll's sample size.

\subsection{Complete Bayesian Simulation Pipeline}

Figure \ref{fig:flowchart} visualizes the complete pipeline. Given a poll of an RCV election,
\begin{enumerate}
    \item Run a bootstrap to determine the pruning list with 99\% + confidence and prune the ballots. 
    \item Aggregate the sample over the unique ballot types $\mathcal{P}$
    \item Impose the prior, which will typically be a flat prior equal to one for each possible ballot ordering in $\mathcal{P}$. If $|\mathcal{P}|$ is large, follow the guidance in Section \ref{sec:poll_size}. If prior polling exists, consider using it to set the prior. 
    \item Compute the posterior Dirichlet distribution using the prior and pruned ballots. 
    \item Sample repeatedly (i.e. $B = 500$ or more) times from the posterior Dirichlet distribution. Each draw is a vector of conditional probabilities for all unique ballot types from the poll.
    \item Treating the conditional probabilities as the distribution of ballot types for a simulated electorate, run RCV on each posterior draw.
    \item Report the average number of wins for each candidate over the $B$ samples as the candidate's probability of winning. 
\end{enumerate}

\begin{figure}[!htbp]
    \centering
    \includegraphics[width=1.1\linewidth]{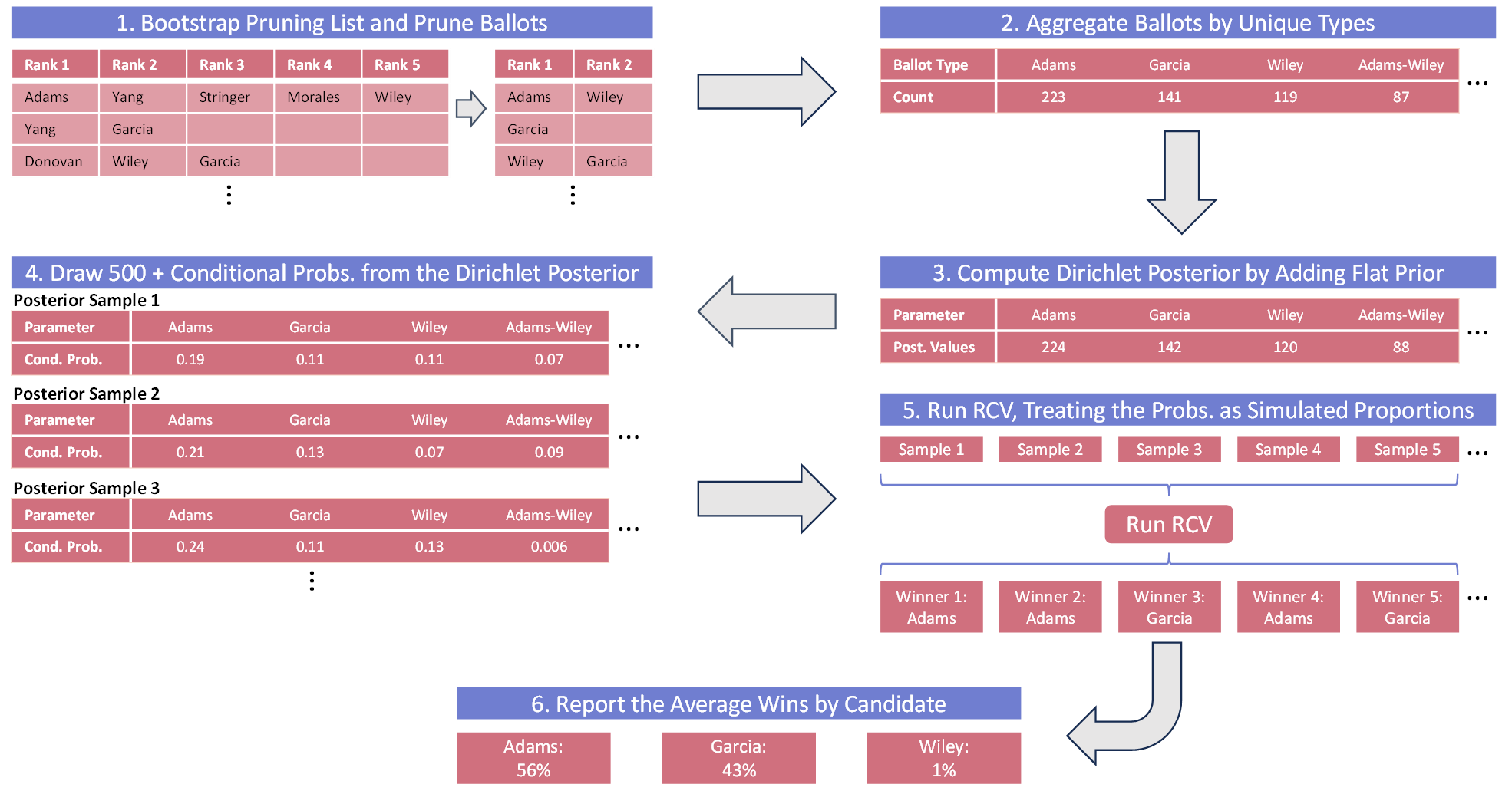}
    \caption{Bayesian Simulation Method Pipeline}
    \label{fig:flowchart}
\end{figure}

\section{Applications to Election Polling Data}
\label{sec:app}

In this section, we apply the proposed Bayesian uncertainty quantification procedure to a real pre-election poll: the final Data for Progress poll of the 2021 Democratic mayoral primary. We also apply the procedure to simulated, unbiased polls drawn from cast vote records. We demonstrate that the method produces interpretable uncertainty summaries for ranked choice voting elections. 

% goal is to demonstrate that the method produces interpretable uncertainty summaries for ranked choice voting elections: posterior probabilities of candidate victory and posterior probabilities of consequential elimination pathways.

\subsection{Reanalyzing the Data for Progress NYC Poll}
\label{subsec:dfp_poll}

We begin with the final Data for Progress poll of the 2021 New York City Democratic mayoral primary. As discussed in Section~\ref{sec:standardMethods}, this poll identified Eric Adams, Kathryn Garcia, and Maya Wiley as the three leading candidates. In the penultimate round, Adams led the poll ($38\%$ support) and Garcia and Wiley were nearly tied ($31\%$ support each). The pollsters then ran scenarios eliminating either Wiley or Garcia. When Garcia was eliminated, Adams won the final round against Wiley at a margin of $54\%$ to $46\%$. When Wiley was eliminated, Garcia clinched a narrow victory over Adams that was well within the poll's margin of error. 

We first use the bootstrap pruning procedure from Section~\ref{sec:pruning_section} to assess the stability of the relevant candidate set. Table~\ref{tab:dfp_pruning} reports the bootstrap agreement rate for pruning sets of different sizes. We see that the poll recovers the same dynamic we saw when sampling from the CVR: the stable positions of Garcia, Adams, and Wiley in the top three, and the stable position of Yang in fourth. Using our heuristic level of 99\%, we can safely prune the poll to just the top three candidates: Adams, Garcia, and Wiley.

\begin{table}[ht]
\centering
\begin{tabular}{cr}
\toprule
\shortstack{Pruning\\Set Size} &
\shortstack{Bootstrap\\Agreement} \\
\midrule
2 & 48.4\% \\
3 & 99.7\% \\
4 & 100.0\% \\
5 & 96.2\% \\
6 & 69.3\% \\
7 & 59.8\% \\
8 & 58.2\% \\
9 & 100.0\% \\
\bottomrule
\end{tabular}
\caption{Percentage of bootstrap samples from the Data for Progress poll yielding the same pruning set at each size. Agreement exceeds 99\% for pruning sets of size three and four, indicating that the set of viable candidates is highly stable.}
\label{tab:dfp_pruning}
\end{table}

The Data for Progress poll responses were weighted by age, gender, education, race, borough, and voting history to reflect the likely electorate. Because the units $i$ have different weights $w_i$, we make a small adaptation to the Bayesian uncertainty from Section~\ref{sec:bayes}. We can define 
\[ n_j^\star = \sum_{i = 1}^n \mathbbm{1}(y_i = p_j) \cdot w_i \qquad \text{for} \qquad j = 1, \dots, |\mathcal{P}|\] 
and then write a pseudo-posterior obtained by replacing counts with weighted counts,  
\[ \boldsymbol{\pi}  \mid (n_1^\star,\ldots,n_{|\mathcal P|}^\star) \sim \mathrm{Dir}\left(\alpha_1 + n_1^\star, \dots, \alpha_{|\mathcal{P}|} + n_{|\mathcal{P}|}^\star\right).\] 
Drawing 500 samples from this posterior distribution, we obtain the posterior win probabilities reported in Table~\ref{tab:dfp_posterior_winner}. 

\begin{table}[h]
\centering
\begin{tabular}{lr}
\toprule
Candidate & Win probability \\
\midrule
Eric Adams & 56.4\% \\
Kathryn Garcia & 42.8\% \\
Maya Wiley & 0.8\% \\
\bottomrule
\end{tabular}
\caption{Posterior win probabilities from the Data for Progress poll after pruning to the top three candidates.}
\label{tab:dfp_posterior_winner}
\end{table}

These probabilities provide a direct summary of election uncertainty. They essentially rule out any possibility of a Maya Wiley victory, which makes sense in context: even if Wiley were to survive to the final round, it was clear that Adams had significantly greater support in a one-on-one matchup. The probabilities also indicate that the final victor is contested between Adams and Garcia -- a highly reasonable conclusion in an election in which Adams defeated Garcia by only 0.8\% in the final round. 

As a gut check, we also sample $25,000$ unbiased polls from the cast vote record, each of size $n = 1,354$, and compute the winner within each poll. We then assess the win frequencies across the unbiased polls, obtaining win ``probabilities" of 68\% for Adams, 32\% for Garcia, with a negligible ($< 0.05$\%) set of samples favoring Wiley. While these numbers do not perfectly align with our posterior probabilities, they are reasonably close and tell the same qualitative story: a two-person race, in which Adams is at most modestly favored. 

\subsection{Weighted Polls Sampled from Cast Vote Records}\label{sec:wtdPolls}

We next evaluate the method using polls randomly sampled from the cast vote records in New York City and Alaska. However, unlike in prior sections of this manuscript, we apply geographic weighting to each of the sampled polls. This is done so that we can more realistically approximate the behavior of real pollsters, who typically weight surveys to the proportions of an expected electorate across dimensions such as partisanship, age, race, and geography. We do not have access to demographic individual-level features in the cast vote record (as such data would risk privacy protections and thus the secret ballot). However, we do have access to geographic information. 

Using this information, we weight the New York City election data by borough. Targets are based on the number of Democratic votes cast in the 2020 presidential election by borough, which yield proportions of 15\% for the Bronx, 30\% for Brooklyn, 26\% for Manhattan, 25\% for Queens, and 4\% for Staten Island. This is a reasonable heuristic, but differs slightly from the true proportions of ballots actually cast in the primary (11\% in the Bronx, 32\% Brooklyn, 27\% Manhattan, 35\% Queens, 4\% Staten Island). Hence, while the underlying samples are unbiased, we induce a small amount of bias via the weighting while simultaneously reducing variance significantly. 

Similarly, we divide Alaska into six geographic regions, based on the geographic divisions used in a 2025 Data for Progress poll \citep{DataForProgress2025PeltolaGov}. We assign weights based on votes cast in the 2020 presidential election, yielding proportions of 39\% for Anchorage, 13\% for the Fairbanks North Star Borough, 9\% for the Kenai Peninsula, 16\% for the Matanuska-Susitna Valley, 12\% for Southeast Alaska, and 11\% for the rest of the state. Again, these proportions differ slightly from what was observed in the actual election (41\%, 12\%, 11\%, 16\%, 11\%, and 9\%, respectively), but represent a reasonable pollster heuristic which induces a small amount of bias while constraining variance. 

As in the prior section, we prune the New York City election to only the top three candidates: Adams, Garcia, and Wiley. The Alaska election only involved three candidates and hence does not require pruning. For each election and sample size, we draw $1,000$ random samples from the cast vote records at each sample size $500, 750, \dots, 2,500$. For each draw, we weight the samples according to the geographic targets, and then compute Dirichlet posterior using our Bayesian method. We then draw $1,000$ samples from the posterior and compute the winner in each sample, allowing us to assign win probabilities -- e.g. Peltola 71\%, Palin 15\%, Begich 14\% -- for each draw. 

Results are visualized in Figure \ref{fig:posterior_win_probs}. In each plot, we use a solid line to show the median win probability across the $1,000$ CVR draws. The darker bands show the interquartile range for win probabilities across the draws. The lighter bands show the middle 80\% range. In each election, we can see that posterior mass moves toward the election winner as the sample size increases, while the width of possible posterior win probabilities narrows. While our simulations show that an ``unlucky" sample could plausibly yield little discrimination between the candidate win probabilities at large survey sizes, we note that additional demographic weighting -- which we cannot simulate from the CVR data -- would presumably narrow the uncertainty bands further. 

\begin{figure}[htbp]
\centering
\includegraphics[width=0.8\textwidth]{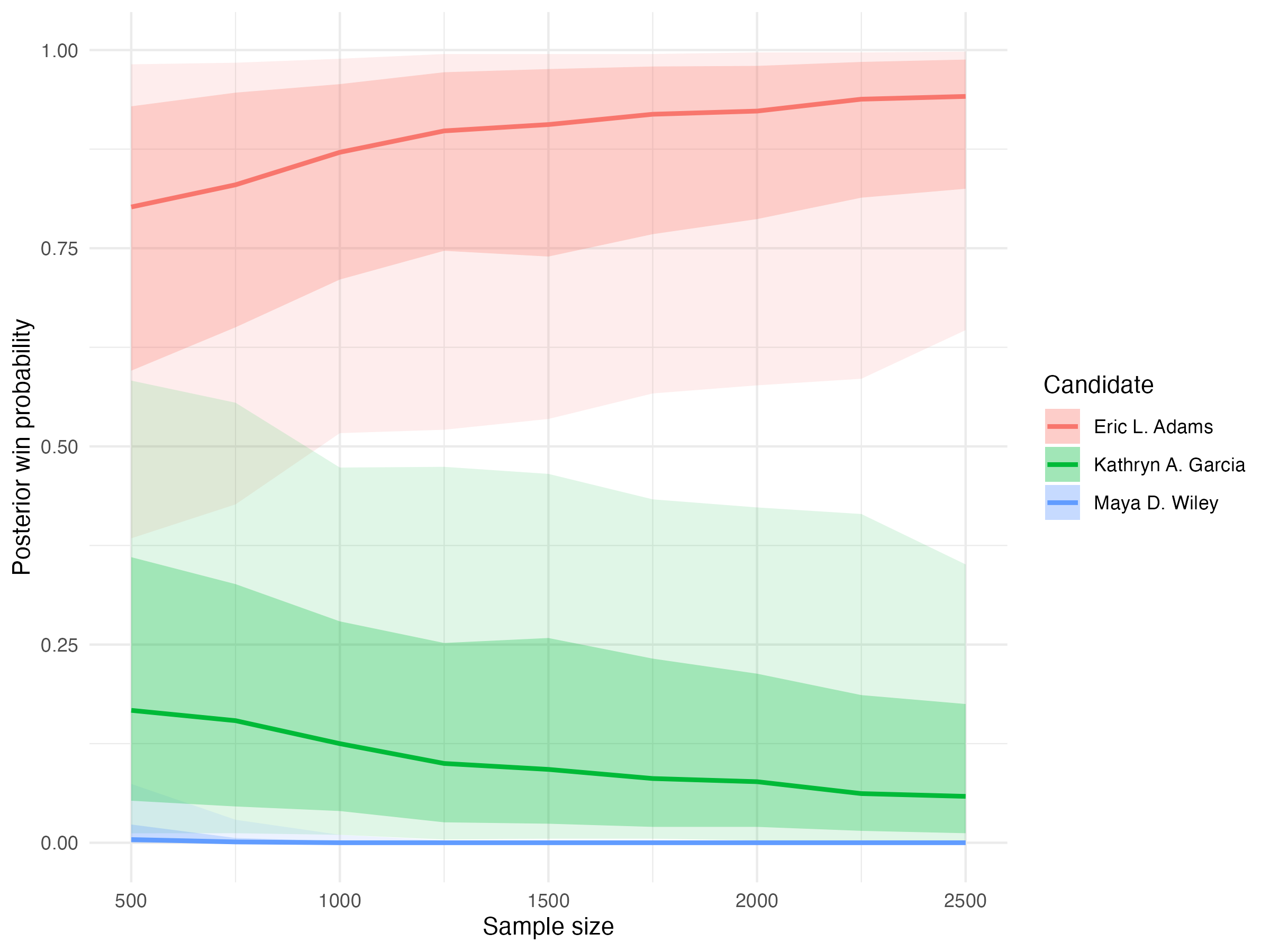}
\vspace{0.5em}
\includegraphics[width=0.8\textwidth]{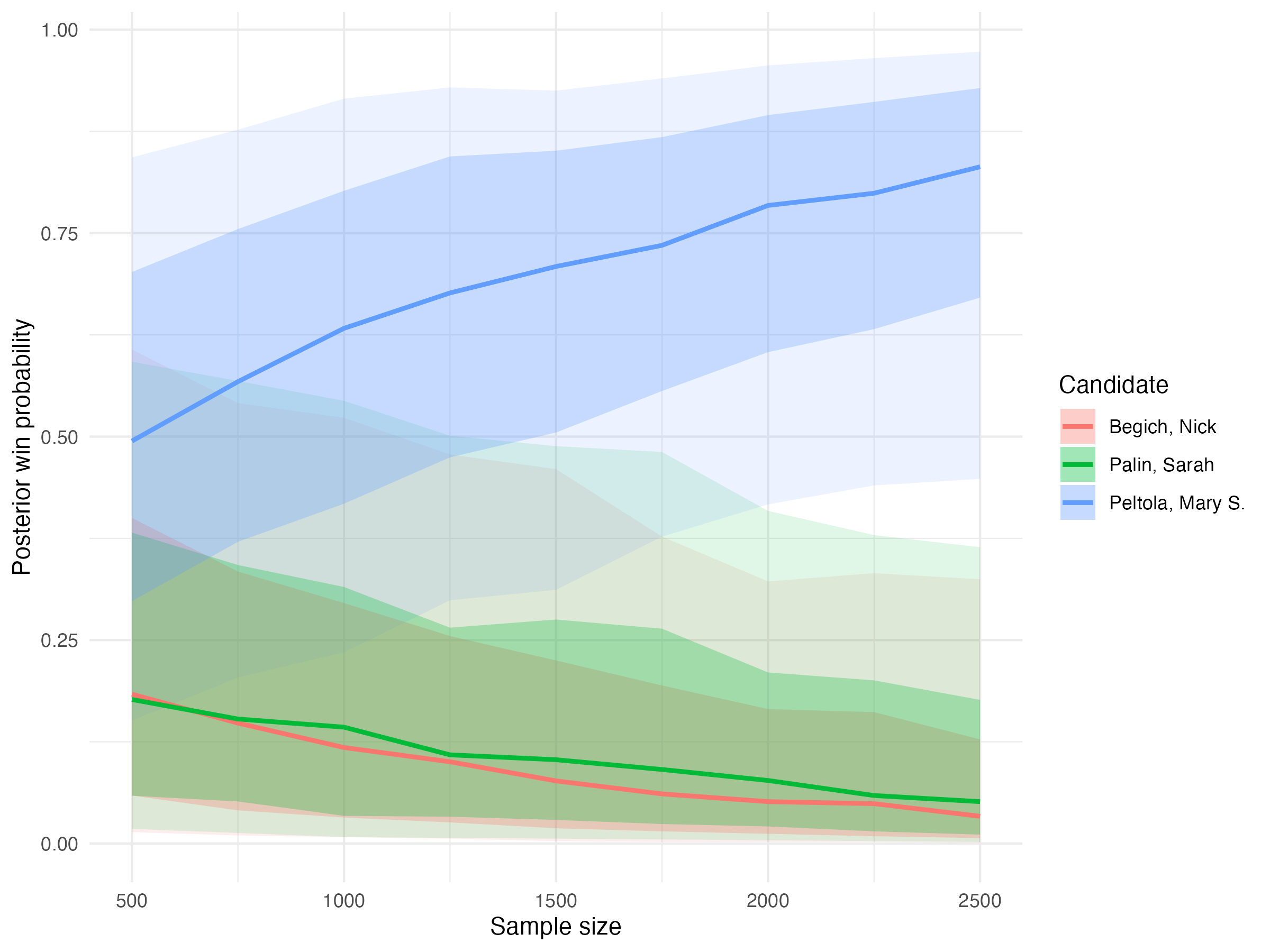}
\caption{Posterior probability of victory as a function of sample size under the Bayesian polling framework. The top panel shows results for the 2021 New York City Democratic mayoral primary, while the bottom panel shows results for the 2022 Alaska special election for U.S. House. Solid lines correspond to the median posterior win probabilities (derived from $1,000$ posterior samples) across $1,000$ random samples from the cast vote records of each election. Darker bands correspond to the IQR of win probabilities, and lighter bands to the middle 80\% range. }
\label{fig:posterior_win_probs}
\end{figure}

\section{Discussion} \label{sec:discussion}

In this manuscript, we have considered uncertainty quantification in the increasingly prevalent setting of ranked choice voting polls. The path-dependent nature of RCV eliminations induces various challenges in assessing who will win an RCV election and how confident we can be in that assessment. We argue that RCV polls are fraught settings for point estimation, because an incorrect winner can actually be favored in a plurality of random samples of a given sample size. Moreover, we note that standard approaches based in computing a poll-level margin-of-error are inadequate, as they only allow us to assess that an estimated elimination pathway is fragile -- not that a given winner is plausible or implausible. We also explore the geometry of the RCV problem, finding that winner assessment involves determining whether a high-dimensional parameter crosses any decision boundaries defined by the unions of polyhedral regions. Because only a few margins are typically determinative, even in high dimensions, standard procedures exhibit low power in assessing RCV winners. 

These deficiencies lead us to seek a different uncertainty measure. A Bayesian approach, using the multinomial-Dirichlet conjugate pair, offers a natural solution because it directly targets the parameter of interest -- the RCV winner -- and is computationally tractable. The approach models the distribution of unique ballot types among an electorate. By repeatedly sampling over the posterior, we compute the average win proportions for candidates across many simulated electorates. These proportions can be interpreted as probability statements about each candidate's election victory likelihood. 

The estimated win probabilities have many natural advantages. They naturally ``bake in" the path-dependent nature of RCV without requiring specialized machinery to handle the underlying geometry. Moreover, unlike traditional MoEs, which only allow binary conclusions on the statistical significance of a poll result, the estimated win percentages do not constrain us to 95\% CIs. Win probabilities offer additional insight on the relative likelihood of candidate victories, even if their leads fall within the MoE. By reporting the estimated win probabilities as opposed to whether the leads exceed a threshold, we offer readers greater context to interpret the poll. However, it is important to communicate to the general electorate that high probabilities are not assurances of victory.

Assessments on real datasets underscore the feasibility and interpretability of this method. For example, in analyzing the final Data for Progress poll of the 2021 New York City Democratic Mayoral Primary, we found estimated win probabilities of $56\%$ for Eric Adams and $43\%$ for Kathryn Garcia. This result closely mirrored the race's true uncertainty: the final-round margin came down to just 0.8\%, and the AP did not call the race for two weeks after the election. A significant advantage of the Bayesian approach is that simple win probabilities can be reported and understood by voters at a glance without requiring additional context.

There are many future extensions to this work. While we find favorable results in two heavily contested and hard-to-poll elections, future work might consider how this approach fares in other settings, such as RCV elections that involve incumbents or in which there is a clearer frontrunner. We have also not considered how this approach might generalize to the case of proportional ranked choice voting elections -- such as those used to elect the Portland, Oregon and Cambridge, Massachusetts city councils --  in which ranked ballots are used to elect multiple candidates, rather than a single winner.  Finally, while the Bayesian framework is most natural to model the distribution of ballots and to incorporate prior polling information, exploratory analysis shows that computational methods such as bootstrapping often yield similar results. Further investigation into computational methods, or specialized frequentist testing procedures, might also help us to address the complexities of uncertainty quantification in ranked choice voting polls. 

\bibliographystyle{apalike2}
\bibliography{rcv_clean}

\newpage

\appendix

\section{Classical Frequentist Inference for Two-Candidate FPTP Polls}
\label{app:fptp}

For completeness, we review the technical details of the standard frequentist framework used to quantify uncertainty in polls of two-candidate first-past-the-post (FPTP) elections.

Let $\pi_A$ and $\pi_B$ denote the population support levels for candidates $A$ and $B$, and define the election margin
\[ d_{AB} = \pi_A - \pi_B = 2\pi_A - 1. \]
Among survey respondents, we encode 
\[ y_i = \begin{cases} 1 & \text{if respondent } i \text{ supports candidate } A,\\ 0 & \text{if respondent } i \text{ supports candidate } B. \end{cases} \]
Under the assumption that $n \ll N$, individual survey responses are approximately independent. Hence, we can model the samples as independent and identically distributed Bernoulli random variables, 
\[ y_1,\ldots,y_n \stackrel{\mathrm{iid}}{\sim} \mathrm{Bernoulli}(\pi_A), \]
and the sample estimate of support for candidate $A$ is
\[ \hat{\pi}_A=\frac{1}{n}\sum_{i=1}^n y_i. \]

As long as $n$ is sufficiently large and $\pi_A$ is sufficiently far from $0$ and $1$, the Central Limit Theorem gives us an approximate normal sampling distribution on $\hat \pi_A$, 
\[ \hat{\pi}_A \sim \mathcal N \!\left(\pi_A,\frac{\pi_A(1-\pi_A)}{n}\right). \]
Applying the linear transformation $\hat d_{AB}=2\hat\pi_A-1$ yields
\[ \hat d_{AB} \sim  \mathcal N\!\left(d_{AB},\frac{4\pi_A(1-\pi_A)}{n}\right). \]
Since $\pi_A(1-\pi_A)\le 1/4$, we typically use the conservative approximation,
\begin{equation}\label{eq:normApprox}
\hat d_{AB} \sim  N\!\left(d_{AB},\frac1n\right), 
\end{equation}
as described in Section \ref{sec:stan}. 

This approach can also be straightforwardly inverted to the problem of testing for the winner of the election, via the one-sided test of hypothesis given by
\[ H_0:d_{AB}\le0 \qquad \text{vs.} \qquad H_1:d_{AB}>0. \]
Namely, using the normal approximation from \eqref{eq:normApprox}, rejection occurs whenever
\[ \sqrt n\,\hat d_{AB} > z_{1-\alpha}. \]
The following result shows that this procedure is uniformly most powerful under the normal approximation.

\begin{lemma}[FPTP Polls as Approximate UMP Tests]\label{rep:fptp}
Suppose we test 
\[ H_0:d_{AB}\le0 \qquad \text{vs.} \qquad H_1:d_{AB}>0, \]
 and assume the normal approximation
\[ \hat d_{AB} \sim \mathcal{N}(d_{AB}, 1/n). \]
Then, the test which rejects when 
\[ \sqrt{n} \times \hat d_{AB} \geq z_{1 - \alpha} \]
is the uniformly most powerful (UMP) test under the normal approximation.
\end{lemma}

\begin{proof}
Denote as $\mathcal{L}(d' \mid \hat d_{AB})$ the likelihood evaluated at $d_{AB} = d'$. Consider two candidate values $d' < d^\dag$ for $d_{AB}$. Then, under our normality assumption, 
\[ \frac{\mathcal L(d^\dag \mid \hat d_{AB})}{\mathcal L(d' \mid \hat d_{AB})} = \exp\left( n(d^\dag-d')\hat d_{AB} -\frac{n((d^\dag)^2-(d')^2)}{2} \right). \]
This expression is monotone in $\hat d_{AB}$. Hence, by the Karlin--Rubin theorem \citep[Theorem 3.4.1 in][]{lehmann2005testing}, there exists a UMP test of the given structure, with the threshold selected to preserve the Type I error rate. Since, by assumption, 
\[ \sqrt{n} \times \hat d_{AB} \sim \mathcal{N}(\sqrt{n} \cdot d_{AB}, 1), \]
the selection of $z_{1-\alpha}$ as the critical value satisfies the Type I error bound. Hence, under the normal approximation, the test is UMP.
\end{proof}

Along with Lemma \ref{rep:fptp}, these results demonstrate the utility of the traditional frequentist inferential framework in simple, two-candidate FPTP races. The problem reduces to testing the value of a single parameter, whose distribution is easy to approximate, allowing for straightforward UMP hypothesis testing as well as confidence interval construction. 

\section{Modeled Margins of Error}\label{sec:mmoe}

In practice, pollsters commonly weight poll responses to effectively oversample from portions of the electorate that are known to have low poll response rates. See \cite{mercer_for_2018} for a complete discussion of poll weighting. Weighting breaks the assumption that the responses are sampled i.i.d from a Bernoulli distribution. Pollsters sometimes quantify the impact of weighting on the sample variance of the poll with the design effect, denoted $D_{\text{eff}}$, which is estimated under Kish's approximation \citep{kish1995methods} to be 
\[
D_{\text{eff}} = 1 + CV(w)^2
\]
where $CV(w)$ is the coefficient of variation of the sample weights. The design effect is then used to compute the effective sample size
\[
n_{\text{eff}} = \frac{n}{D_{\text{eff}}}
\]
which can be substituted into the formula for the standard error using the conservative variance estimator. Under the Pew Research standard, these uncertainty measures are reported as modeled margins of error (mMoEs) See \cite{mercer_understanding_2016} for further discussion of the design effect. Conservative survey-level margins of error--calculated with either the true sample size or effective sample size due to weighting--form the basis for frequentist measures of polling uncertainty.

\section{Proof of Lemma \ref{lemma:numParams}}\label{proof:paramsLemma}
\begin{proof}
We label the candidates $1, \dots, K$ and write $\mathcal{S} = \{1, \dots, K\}$. To assess the winner of an RCV election, observe that we need to know the support level of each candidate at each elimination round, and we impose no restrictions on how support flows from one candidate to another. Moreover, we do not know a priori which candidate will be eliminated in each round. 

Denote as $2^\mathcal{S}$ the power set of $\mathcal{S}$ such that $S \in 2^\mathcal{S} \implies S \subseteq \mathcal{S}$. Every entry $S \in 2^\mathcal{S}$ such that $|S| \geq 2$, corresponds to a possible set of remaining candidates at the $(K - |S| + 1)^{th}$ elimination round. Define the ``support level" for each candidate in a given set $S$ to mean the percentage of ballots listing that candidate first in an elimination round in which the candidates in $S$ are those who are not yet eliminated. To assess who is eliminated in each round, we need to know the support level for all but one of the candidates (as the final candidate's support level is constrained by the restriction that the support levels add up to 1), such that the minimum value can be identified. Moreover, because we impose no restrictions on how support is redistributed following candidate eliminations, the support vector on a set $S$ may vary freely over the $(|S|-1)$-dimensional simplex. Consequently, no smaller collection of free parameters can characterize all possible support configurations on $S$.

Lastly, observe that there are ${K \choose s}$ subsets $S \subseteq \mathcal{S}$ which have cardinality $s$. Hence, mathematically, we need 
\[ \sum_{s=2}^K {K \choose s}(s-1) \]
parameters to identify the winner. 
Using the binomial identities
\[ \sum_{s=0}^K s {K \choose s} = K2^{K-1} \qquad \text{and} \qquad \sum_{s=0}^K {K \choose s} = 2^K, \]
we obtain
\[ \sum_{s=2}^K {K \choose s}(s-1) = \sum_{s=0}^K {K \choose s}(s-1) + 1 = K2^{K-1} - 2^K + 1 = 2^{K-1}(K-2)+1.\]
Lastly, note that we have expressed the relevant parameters as ``support levels." However, we can always choose one candidate in each set $S$ as the ``reference" candidate, and rewrite the relevant parameters as the difference between that candidate's support level and every other candidate's support level. Hence, we can parameterize the entire problem in terms of margins rather than support levels. 
\end{proof}

\section{Review: Bayesian Statistics} \label{sec:bayes}

This section is based on \cite{bolstad_introduction_2007}, which offers a thorough introduction to Bayesian statistics. The Bayesian approach to statistical inference considers a population parameter to be a random variable whose distribution is estimated using prior knowledge and observed data. Unlike under the frequentist approach, this allows us to make direct probability statements on the value of the parameter. For the simplicity, we consider only discrete probability distributions in this section. Let $A$ be the finite universe of parameter values, and let $a \in A$ be a possible parameter value. To represent our prior knowledge about the data, a prior probability mass function (PMF) is first assumed on $A$. Often, no meaningful assumptions can be made about the parameter beforehand. In this case, a flat, uninformative prior which assigns an equal probability to all $a \in A$ is used. 

Separately, the likelihood of observing the survey data, $y_{\text{data}}$, under different possible values of $a$ is considered. The likelihood function, $\mathcal{L}$, is defined for all $a$ by
\[
\mathcal{L}(a; y_{data}) = f(y_{\text{data}} \mid a) = \mathbb{P}(Y = y_{\text{data}} \mid a )
\]

This is the PMF of the observed data evaluated at $y_{\text{data}}$ viewed as a function of $a$. The likelihood function has the same underlying form the relevant PMF under which the sample data is assumed to be sampled from but with the data instead of the parameter value fixed. By considering the values of $\mathcal{L}(a; y_{data})$ across possible parameter values, we obtain a set of relative weights that indicate how well each $a$ explains the observed data. For example, if 
\[ f(a_i \mid y_{\text{data}}) = 0.1 \quad \text{and} \quad f(a_j \mid y_{\text{data}}) = 0.05\]
then we say that the sample data suggests the parameter is twice as likely to to be $a_i$ than $a_j$. 

The likelihood function is then reconciled against the prior distribution to obtain the posterior distribution of the parameter. Bayes' theorem says that the posterior distribution is proportional to the product of the prior distribution and the likelihood: 
\[
posterior \propto prior \times likelihood
\]
For each $a$, the prior probability is weighted by the likelihood and rescaled to obtain a posterior probability for $a$. Parameter values with high likelihood have increased posterior probabilities, and parameters that are unlikely according to the data have their probabilities reduced. Direct probability statements can then be made on the parameter value via the posterior distribution. 

Many likelihood functions from known probability distributions exhibit a conjugacy property \citep{bolstad_introduction_2007}. This means that the posterior will come from the same family of probability distributions as the prior distribution. Conjugate posteriors can be easily obtained by combining the prior distribution and the likelihood function under simple update rules defined for each conjugate family.

\section{Proof of Lemma \ref{pruneThm} }\label{sec:pruneProof}

\begin{proof}
We show that sequential elimination of minor candidates results in the same distribution of ballots -- and therefore RCV vote counts -- as the batch elimination of a group of minor candidates at once. Suppose $\mathcal C$ is the set of all candidates. Let $\mathcal{C}_m^\star$ be the pruning set of major candidates and $E$ be the set of remaining minor candidates so that $\mathcal C = \mathcal{C}_m^\star \cup E$ and $\mathcal{C}_m^\star \cap E = \emptyset$. Consider an individual ballot $b = (b_1, \ldots, b_R)$ where $R$ is the maximum number of candidates a ballot can rank. Then

\[
b_i \in \mathcal{C}_m^\star \cup E \quad \text{or} \quad b_i = \varnothing \quad \forall i \in \{1, \ldots, R\}
\]
where $b_i = \varnothing$ if there is a blank in the $i^{\text{th}}$ position. 

Let the candidates in $E$ be sequentially eliminated according to RCV mechanics so that the candidate eliminated in the $r^{\text{th}}$ round is $e_r\in E$. Then if $b_i = e_r$ for some $i$, $b$ is updated by replacing $b_i = \varnothing$. 

After $\vert E \vert$ rounds, $b$ contains only major candidates and gaps or it has been exhausted $(\text{if} \hspace{0.5em} b_i = \varnothing \hspace{0.5em} \forall i)$ and is discarded from the set of active ballots. That is, we obtain the result 

\[ b_i \in \mathcal{C}_m^\star \quad \text{or} \quad b_i = \varnothing \quad \forall i \]
which is exactly the same result as if we set all $b_i = \varnothing$ if $b_i \in E$ at once. 

\end{proof}

\end{document}